\documentclass[11pt,letterpaper]{article}
\usepackage[margin=1in]{geometry}
\usepackage{amsmath,amssymb,amsfonts,setspace,url}
\usepackage{bm}
\usepackage{latexsym}
\usepackage{amsthm}
\usepackage{subfig}
\usepackage{float}
\usepackage{algorithmic}
\usepackage[ruled]{algorithm}
\usepackage{bigstrut,array,multirow}
\usepackage{here}
\usepackage{color}

\theoremstyle{plain}
\newtheorem{theorem}{Theorem}
\newtheorem{lemma}{Lemma}

\theoremstyle{definition}

\newcommand{\Tri}{\textsc{FindTriangleInSubnetwork}}

\newcommand{\Ee}{\mathcal{E}}

\newcommand{\ceil}[1]{\left\lceil #1 \right\rceil}
\newcommand{\floor}[1]{\left\lfloor #1 \right\rfloor}
\providecommand{\Aa}{\mathcal{A}}

\providecommand{\Gin}{G^{\textrm{in}}}
\providecommand{\Vin}{V^{\textrm{in}}}
\providecommand{\Vout}{V^{\textrm{out}}}
\providecommand{\Ein}{E^{\textrm{in}}}
\providecommand{\Einter}{E^{\textrm{inter}}}
\providecommand{\degin}{\textrm{deg}_{\Ein}}
\providecommand{\degout}{\textrm{deg}_{\Eout}}
\providecommand{\Eout}{E^{\textrm{out}}}
\providecommand{\Enew}{E^{\textrm{new}}}
\newcommand{\poly}{\mathrm{poly}}
\newcommand{\mix}{\mathrm{mix}}
\newcommand{\diam}{\mathrm{diam}}

\newcommand{\mym}{\bar{m}}
\newcommand{\myn}{\bar{n}}
\providecommand{\Nn}{\mathcal{N}}
\providecommand{\Vv}{\mathcal{V}}
\providecommand{\Ii}{\mathcal{I}}
\providecommand{\Cc}{\mathcal{C}}
\newcommand{\Hline}[1]{\noalign{\hrule height #1}}

\begin{document}

\title{Quantum Distributed Algorithm for Triangle Finding in the CONGEST Model}
\author{Taisuke Izumi\\
Graduate School of Engineering\\
Nagoya Institute of Technology\\
\url{t-izumi@nitech.ac.jp}
\and
Fran{\c c}ois Le Gall\\
Graduate School of Mathematics\\
Nagoya University\\
\url{legall@math.nagoya-u.ac.jp}
\and
Fr\'ed\'eric Magniez\\
Universit\'e de Paris, IRIF, CNRS\\
\url{frederic.magniez@irif.fr}
}
\date{}

\maketitle
\thispagestyle{empty}
\setcounter{page}{0}
\begin{abstract}
This paper considers the triangle finding problem in the CONGEST model of distributed computing. Recent works by Izumi and Le Gall (PODC'17), Chang, Pettie and Zhang  (SODA'19) and Chang and Saranurak (PODC'19) have successively reduced the classical round complexity of triangle finding (as well as triangle listing) from the trivial upper bound $O(n)$ to $\tilde O(n^{1/3})$, where~$n$ denotes the number of vertices in the graph. In this paper we present a quantum distributed algorithm that solves the triangle finding problem in $\tilde O(n^{1/4})$ rounds in the CONGEST model. This gives another example of quantum algorithm beating the best known classical algorithms in distributed computing. Our result also exhibits an interesting phenomenon: while in the classical setting the best known upper bounds for the triangle finding and listing problems are identical, in the quantum setting the round complexities of these two problems are now $\tilde O(n^{1/4})$ and $\tilde \Theta(n^{1/3})$, respectively. Our result thus shows that triangle finding is easier than triangle listing in the quantum CONGEST model.
\end{abstract}
\newpage

%=================================
\section{Introduction}\label{section:intro}
%=================================
{\bf Background.}
The problem of detecting triangles in graphs has recently become the target of intensive research by the distributed computing community \cite{Abboud+17, CKKLPS15, Chang+SODA19, Chang+PODC19, Dolev+DISC12, Drucker+PODC14, Izumi+PODC17,Pandurangan+SPAA18}. This problem comes in two main variants: the \emph{triangle finding problem} and the \emph{triangle listing problem}. Given as input a graph $G=(V,E)$, the triangle finding problem asks to decide\footnote{Another version of the triangle finding problem asks to output one triangle of $G$ (or report that the graph has no triangle). It is easy to see that the two versions are essentially equivalent: a triangle can be found by applying $O(\log |V|)$ times an algorithm solving the decision version.}  whether the graph contains a triangle (i.e., three vertices $u,v,w\in V$ such that $\{u,v\}, \{u,w\},\{v,w\}\in E$), while the triangle listing problem asks to output all the triangles of $G$. A solution to the latter version, obviously, gives a solution to the former version. Besides its theoretical interest, another motivation for considering these problems is that for several graph problems faster distributed algorithms are known over triangle-free graphs (e.g., \cite{Hirvonen+17,Pettie+15}). The ability to efficiently check whether the network is triangle-free (or detect which part of the network is triangle-free) is essential when considering such algorithms. 

One of the main models to study graph-theoretic problems in distributed computing is the CONGEST model. In this model the graph $G = (V, E)$ represents the topology of the network, the computation proceeds with round-based synchrony and each vertex can send one $O(\log n)$-bit message to each adjacent vertex per round, where $n$ denotes the number of vertices. Initially, each vertex knows only the local topology of the network, i.e., the set of edges incident to itself. The triangle finding and listing problems ask, respectively, to decide if $G$ contains a triangle and to list all triangles of~$G$. The trivial strategy is for each vertex to send the list of all its neighbors to each neighbor (all the triangles can then be listed locally, i.e., without further communication). Since each list can contain up to $n$ vertices, this requires $O(n)$ rounds in the CONGEST model.

In 2017, Izumi and Le Gall \cite{Izumi+PODC17} gave the first nontrivial distributed algorithms for triangle detection in the CONGEST model: they constructed a $\tilde O(n^{2/3})$-round randomized algorithm\footnote{In this paper the notations $\tilde O(\cdot)$ and $\tilde \Omega (\cdot)$ remove $\poly(\log(n))$ factors.} for triangle finding and a $\tilde O(n^{3/4})$-round randomized algorithm for triangle listing. This was soon improved by Chang, Pettie and Zhang~\cite{Chang+SODA19}, who obtained a $\tilde O(\sqrt{n})$-round randomized algorithm for both the triangle finding and listing problems. The key idea leading to this improvement was to decompose the graph into components with low mixing time, and then apply recent routing techniques \cite{Ghaffari+PODC17,Ghaffari+DISC18} that make possible to achieve efficient routing in graphs with low mixing time. The complexity of the resulting distributed algorithm was dominated by the cost required to compute the graph decomposition. Very recently, Chang and Saranurak~\cite{Chang+PODC19} developed a much more efficient method to decompose the graph into components with low mixing time, which immediately leads to $\tilde O(n^{1/3})$-round randomized algorithms for the triangle finding and listing problems. Since a matching lower bound $\tilde \Omega(n^{1/3})$ is known for triangle listing \cite{Pandurangan+SPAA18}, the randomized round complexity of the triangle listing problem is thus now settled, up to possible polylogarithmic factors.\footnote{An interesting open problem, however, is to determine the deterministic round complexity of these problems. To our knowledge, no deterministic algorithm with sublinear round complexity is known in the CONGEST model.} For the triangle finding problem, on the other hand, essentially no nontrivial lower bound is known. Two exceptions are, first, the very weak lower bound obtained by Abboud, Censor-Hillel, Khoury and Lenzen \cite{Abboud+17} in the CONGEST model and, second, a lower bound obtained by Drucker, Kuhn and Oshman~\cite{Drucker+PODC14} for the much weaker CONGEST-BROADCAST model (where at each round the vertices can only broadcast a single common message to all other vertices) under a conjecture in computational complexity theory. This leads to the following intriguing question: is triangle finding easier than triangle listing?

Another related, but much stronger, model is the CONGEST-CLIQUE model. In this model at each round messages can even be sent between non-adjacent vertices, which makes bandwidth management significantly easier. In the CONGEST-CLIQUE model, Dolev, Lenzen and Peled~\cite{Dolev+DISC12} first showed that the triangle listing problem (and thus the triangle finding problem as well) can be solved deterministically in $\tilde O(n^{1/3})$ rounds for general graphs, which is tight since
%, where $n$ denotes the number of vertices in the graph, and in $O(\Delta^{3}/n)$ rounds for graphs of maximum degree $\Delta$. 
the lower bound $\tilde \Omega(n^{1/3})$ by Pandurangan, Robinson and Scquizzato~\cite{Pandurangan+SPAA18} mentioned above holds in this model as well. Note that this lower bound also means that triangle listing in the CONGEST-CLIQUE model is not easier than triangle listing in the CONGEST model, at least as far as randomized algorithms are concerned. For the triangle finding problem, on the other hand, the better upper bound $O(n^{0.1572})$ has been obtained by Censor-Hillel et al.~\cite{CKKLPS15} by implementing fast matrix multiplication algorithms in the distributed setting. The CONGEST-CLIQUE model is thus a setting in which triangle finding is easier than triangle listing.

Table \ref{table:comparison} summarizes the best known bounds on the round complexity of triangle finding and listing discussed so far. \vspace{2mm}

\setlength{\extrarowheight}{1.5pt}
\begin{table*}[t]
 \begin{center}
  \begin{tabular}{lllll}
  \Hline{1.5pt}
   Model& Setting  & Problem &Complexity& Paper  \\ \Hline{1.5pt}
   CONGEST-CLIQUE & Classical & Listing&$\tilde O(n^{1/3})$& Dolev et al.~\cite{Dolev+DISC12}\bigstrut\\ \hline
   CONGEST-CLIQUE & Classical & Finding &$O(n^{0.1572})$&Censor-Hillel et al.~\cite{CKKLPS15}\bigstrut\\ \hline
%   Le Gall\cite{LeGall16} & $\tilde{O}(n^{1-2/\omega}) = \tilde{O}(n^{0.1572}) $ &
%   Counting & CONGEST clique \\ \hline 
   CONGEST & Classical & Listing & $\tilde O(n^{1/3})$ &Chang and Saranurak~\cite{Chang+PODC19} \bigstrut\\ \hline
   CONGEST & Quantum & Finding &$\tilde O(n^{1/4})$ & This paper \bigstrut\\  \Hline{1.5pt}

%    \multirow{2}{*}{CONGEST broadcast}& \multirow{2}{*}{quantum}& \multirow{2}{*}{Finding} &$\Omega\big(\frac{n}{e^{\sqrt{\log n}}\log n}\big)$&\multirow{2}{*}{Drucker et al.~\cite{Drucker+PODC14}} \bigstrut\\
%   &&&(conditional)& \bigstrut\\ \hline  
   CONGEST-BROADCAST& Classical& Finding &$\Omega\big(\frac{n}{e^{\sqrt{\log n}}\log n}\big)$&Drucker et al.~\cite{Drucker+PODC14} \bigstrut\\ \hline
   CONGEST-CLIQUE& Quantum& Listing &$\Omega\big(\frac{n^{1/3}}{\log n}\big)$&Pandurangan et al.~\cite{Pandurangan+SPAA18} \bigstrut\\ \Hline{1.5pt}
%   \hline
%   Izumi and Le Gall \cite{Izumi+PODC17}& $\Omega(\frac{n^{1/3}}{\log n})$ &quantum&Listing & CONGEST clique \bigstrut\\ \Hline{1.5pt}
  \end{tabular}
  
   \caption{Prior results on the round complexity of distributed triangle finding and listing, and our new result. Here $n$ denotes the number of vertices of the graph. Note that any upper bound for the listing problem (in particular, the upper bound from \cite{Chang+PODC19}) holds for the finding problem as well. Similarly, note that any lower bound for the quantum CONGEST-CLIQUE model (in particular, the lower bound from \cite{Pandurangan+SPAA18}) holds for the weaker classical and quantum CONGEST models as well.}
  \label{table:comparison} \vspace{-7mm}
 \end{center}
\end{table*}

%%%%
\noindent{\bf Quantum distributed computing.}
Quantum versions of the main models studied in distributed computing can be easily defined by allowing quantum information, i.e., quantum bits (qubits), to be sent through the edges of the network instead of classical information, i.e., bits. In particular, in the quantum version of the CONGEST model, which we will simply call the ``quantum CONGEST model'' below, each vertex can send one message of $O(\log n)$ qubits  to each adjacent vertex per round. While a seminal work by 
%Gavoille et al.~\cite{Gavoille+DISC09} and 
Elkin et al.~\cite{Elkin+PODC14} showed that for many important graph-theoretical problems the quantum CONGEST model is not more powerful than the classical CONGEST model, Le Gall and Magniez \cite{LeGall+PODC18} recently showed that one problem can be solved more efficiently in the quantum setting: computing the diameter of the network. More precisely, they constructed a $\tilde{O}(\sqrt{nD})$-round quantum algorithm for the exact computation of the diameter of the network (here~$D$ denotes the diameter), while it is known that any classical algorithm in the CONGEST model requires~$\tilde{\Omega}(n)$ rounds, even for graphs with constant diameter~\cite{FHW12}. In the CONGEST-CLIQUE model as well, a quantum algorithm faster than the best known classical algorithms has been obtained recently for the All-Pair Shortest Path problem~\cite{Izumi+PODC19}.  In the LOCAL model, which is another fundamental model in distributed computing, separations between the computational powers of the classical and quantum versions have also been reported \cite{Gavoille+DISC09,LeGall+STACS19}.

When discussing the classical randomized complexity of triangle listing in the CONGEST and CONGEST-CLIQUE models, we mentioned the $\tilde \Omega(n^{1/3})$-round lower bound by Pandurangan, Robinson and Scquizzato \cite{Pandurangan+SPAA18}. This lower bound is based on an information-theoretic argument that actually holds even in the quantum setting. In view of the recent matching upper bound in the classical setting \cite{Chang+PODC19}, we can conclude that for triangle listing the quantum CONGEST model is not more powerful than the classical CONGEST model. An intriguing question is whether quantum communication can help solving faster the triangle finding problem in the CONGEST model. In particular, can we break the $\tilde O(n^{1/3})$ barrier for triangle finding in the quantum setting?\vspace{2mm}

%%%
\noindent{\bf Our result.}
In this paper we break this barrier. Our main result is the following theorem.
\begin{theorem}\label{th:main}
In the quantum CONGEST model, the triangle finding problem can be solved with probability at least $1-1/\poly(n)$ in $\tilde O(n^{1/4})$ rounds, where $n$ denotes the number of vertices in the network. 
\end{theorem}
In comparison, as already explained, in the classical CONGEST model the best known upper bound on the randomized round complexity of triangle finding is $\tilde O(n^{1/3})$. Our result thus gives another example of quantum algorithm beating the best known classical algorithms in distributed computing.  It also exhibits an interesting phenomenon: while in the classical setting the best known upper bounds for the triangle finding and listing problems are both $\tilde O(n^{1/3})$, in the quantum setting the round complexity of the former problem becomes $\tilde O(n^{1/4})$ while the round complexity of the latter problem remains $\tilde \Theta(n^{1/3})$. Theorem \ref{th:main} thus shows that triangle finding is easier than triangle listing in the quantum CONGEST model.\vspace{2mm}
%This shows for the first time a difference of triangle finding and triangle listing in the CONGEST model.

%%%
\noindent{\bf Overview of our approach.}
Our approach starts similarly to the classical algorithms developed by Chang, Pettie and Zhang~\cite{Chang+SODA19} and Chang and Saranurak~\cite{Chang+PODC19}. As in \cite{Chang+SODA19}, by using an expander decomposition of the network, the triangle finding problem over the whole network is reduced to the task of detecting whether the subnetwork induced by a set of edges $\Ein\cup\Eout$ contains a triangle. We denote the latter problem $\Tri$.  Here $\Ein$ and $\Eout$ are two subsets of edges that satisfy specific conditions. In particular, the subnetwork induced by the edges in~$\Ein$ is guaranteed to have low mixing time (but in general nothing can be said about the mixing time of the subnetwork induced by $\Ein\cup\Eout$). We use the algorithm from \cite{Chang+PODC19} to compute the expander decomposition efficiently, which implies that an efficient algorithm for $\Tri$ gives an efficient algorithm for the triangle finding problem over the whole network. More precisely, a $\tilde O(n^{1/4})$-round algorithm for $\Tri$ leads to a $\tilde O(n^{1/4})$-round algorithm for the triangle finding problem. The details of this reduction are described in Section \ref{section:reduction}.

Our main approach to solve the problem $\Tri$ is to apply the framework for quantum distributed search developed in \cite{LeGall+PODC18}. We briefly sketch the main ideas. Let us write $\Vv\subseteq V$ the set of vertices of the graph induced by the edges in $\Ein\cup \Eout$. We will partition~$\Vv$  into $t=\tilde \Theta(\sqrt{n})$ subsets $\Vv_1,\Vv_2,\ldots,\Vv_t$ each containing $\tilde O(\sqrt{n})$ vertices. Let us write $\Lambda=[t]\times [t]\times [t]$. We will try to find a triple $(i,j,k)\in \Lambda$ for which there exist $u\in \Vv_i$, $v\in\Vv_j$ and $w\in\Vv_k$ such that $\{u,v,w\}$ is a triangle. To do this, we partition the set $\Lambda$ into $t$ subsets $\Lambda_1,\cdots,\Lambda_t$ each containing $t^2=\tilde \Theta(n)$ triples and consider the following search problem: find an index $\ell\in[t]$ such that the set~$\Lambda_\ell$ contains a triple $(i,j,k)$ for which there exist $u\in \Vv_i$, $v\in\Vv_j$ and $w\in\Vv_k$ such that $\{u,v,w\}$ is a triangle. Quantum distributed search enables us to solve this problem in $\tilde O(\sqrt{t}\delta)$ rounds in the quantum CONGEST model if a checking procedure (i.e., a procedure that on input $\ell$ checks if the set $\Lambda_\ell$ contains a triple $(i,j,k)$ for which there exist $u\in \Vv_i$, $v\in\Vv_j$ and $w\in\Vv_k$ such that $\{u,v,w\}$ is a triangle) can be implemented in $\delta$ rounds. 

The checking procedure distributes the $\tilde \Theta(n)$ triples in $\Lambda_\ell$ among the vertices of the network proportionally to the degree of the vertices. In particular, vertices with very low degree do not receive any triple. This technique is essentially the same as for the main procedure in the classical algorithm by Chang, Pettie and Zhang~\cite{Chang+SODA19}. Next, each vertex owning a triple $(i,j,k)$ checks whether there exist $u\in \Vv_i$, $v\in\Vv_j$ and $w\in\Vv_k$ such that $\{u,v,w\}$ is a triangle, which requires gathering information of all the edges with extremities in these sets. Since the subnetwork induced by the edges in $\Ein$ has low mixing time, we would like to use the same classical routing techniques \cite{Ghaffari+PODC17,Ghaffari+DISC18} as used in the main procedure of the classical algorithms~\cite{Chang+SODA19,Chang+PODC19}. Special care is nevertheless required to ensure that we can apply those routing techniques. (This was not needed in \cite{Chang+SODA19,Chang+PODC19} since these prior works used a different approach: each vertex simply loaded all the necessary edges from $\Lambda$ in $\tilde O(n^{1/3})$ rounds. In comparison, we need to guarantee that the necessary edges from $\Lambda_\ell$, for a fixed $\ell$, can be loaded in negligible time.) We solve this difficulty by carefully defining the sets $\Lambda_\ell$ so that the the same edge is not requested by two distinct vertices at the same time (this is the contents of Lemma \ref{lemma:unique} in Section \ref{sub-small}).   

Another technical difficulty is how to handle vertices with very high degree. In the classical setting this was trivial, since it was enough to gather in $\tilde O(n^{1/3})$ rounds all the information about the edges of the network at one of these high-degree vertices. Since we want to construct a $\tilde O(n^{1/4})$-round quantum algorithm we cannot use this approach. Instead, we develop an approach based on the well-known protocol from the two-party quantum computation complexity of the disjointness function \cite{Buhrman+STOC98}. This is explained in Section \ref{sub-large}. \vspace{2mm}
 
\noindent{\bf Other related works.}
The triangle finding problem is also a central problem in quantum query complexity. While many quantum query algorithms have been designed in this setting \cite{BelovsSTOC12,LeGallFOCS14,Lee+SODA13,Magniez+SICOMP07}, they are based on quantum techniques (e.g., quantum walk search and learning graphs) that do not seem to lead to efficient algorithms in the distributed setting.

%=================================
\section{Preliminaries}\label{section:preliminaries}
%=================================
{\bf Graph theory.}
All the graphs considered in this paper are undirected and unweighted.  For any graph $G=(V,E)$ and any vertex $u\in V$, we denote $\deg(u)$ the degree of $u$ and $\Nn(u)$ the set of neighbors of $u$. We write $n=|V|$ and $m=|E|$. For any set $E'\subseteq E$, we denote $\deg_{E'}(u)$ the number of edges in $E'$ incident to $u$ and write $\Nn_{E'}$ the set of all neighbors $v\in \Nn(u)$ such that $\{u,v\}\in E'$. We denote $\diam(G)$ the diameter of $G$ and $\mix(G)$ the mixing time of~$G$, i.e., the number of steps of a random walk over $G$ needed to obtain a distribution close to the stationary distribution (we refer to \cite{Ghaffari+PODC17} for a precise definition). For any positive integer $t$, we write $[t]=\{1,2,\ldots,t\}$.

We will use the following lemma from \cite{Chang+SODA19} in our main algorithm.
\begin{lemma}[Lemma 4.2 in \cite{Chang+SODA19}]\label{lemma}
Consider a graph with $m$ edges and $n$ vertices. Let $p$ be such that $p^2m\ge 400(\log n)^2$.
Suppose that the degree of any vertex of the graph is at most $mp/(20\log n)$.  Generate a subset $S$ by letting each vertex join $S$ independently with probability $p$. Then with probability at least $1-1/\poly(n)$, the number of edges in the subgraph induced by $S$ is at most $6p^2m$.
\end{lemma}

\noindent{\bf Classical distributed computing.}
In the classical CONGEST model, the graph $G = (V, E)$ represents the topology of the network. The computation proceeds with round-based synchrony and each vertex can send one $O(\log n)$-bit message to each adjacent vertex per round. All links (corresponding to the edges of $G$) are reliable and suffer no faults. Each vertex has a distinct identifier from a domain $\Ii$ with $|\Ii| = \mathrm{poly}(n)$. It is also assumed that each vertex can access an infinite sequence of local random bits. Initially, each vertex knows nothing about the topology of the network except the set of edges incident to itself and the value of $n$. 

%It will be convenient to introduce the following concept of \emph{active edges}.
%%\begin{definition}
%Let $E'\subseteq E$ be a subset of edges of the network. Solving a computational problem in the setting where $E'$ is the set of active edges means solving the problem by an algorithm that only sends messages along the edges in $E'$ (i.e., the communication is restricted to the subnetwork induced by $E'$).
%%\end{definition}

Our quantum algorithm will be based on several classical distributed algorithms from the literature. A first crucial ingredient is the following recent result by Chang and Saranurak~\cite{Chang+PODC19} that shows how to efficiently compute a good expander decomposition of the graph.
\begin{theorem}[\cite{Chang+PODC19}]\label{th:decomposition}
In the classical CONGEST model, there is a $O(n^{0.1})$-round algorithm that computes with probability at least $1-1/\poly(n)$ a partition 
\[
V=V_1\cup V_2\cup\cdots\cup V_s
\] 
of the vertex set $V$ that satisfies the following two conditions:
\begin{itemize}
\item
for each $i\in[s]$, the subgraph induced by the vertex set $V_i$ has mixing time
$O(\poly(\log n))$;
\item
the number of inter-component edges (i.e., the number of edges with one endpoint in $V_i$ and one endpoint in $V_j$, for $i\neq j$) is at most $|E|/10$.
\end{itemize}
\end{theorem}

We will also use the following technical lemma from \cite{Chang+SODA19} that shows how to compute efficiently a new ID assignment  that gives a good estimation of the degree of any vertex of the graph.
\begin{lemma}[Lemma 4.1 in \cite{Chang+SODA19}]\label{th:ID}
In the classical CONGEST model, there is a $O(\diam(G)+\log n)$-round deterministic algorithm that computes a bijective map $\gamma\colon V\to\{1,\ldots,|V|\}$ and a function $d\colon\{1,\ldots,|V|\}\to \{0,1,\ldots,\floor{\log_2(n)}\}$ satisfying the following conditions: 
\begin{itemize}
\item[(i)]
$\gamma(u)\le \gamma(v)$  implies $\floor{\log_2(\deg(u))}\le \floor{\log_2(\deg(v))}$ for any $u,v\in V$;
\item[(ii)]
$d(\gamma(u))=\floor{\log_2(\deg(u))}$ for all $u\in V$.
\end{itemize}
More precisely, after running the algorithm each vertex $u$ knows $\gamma(u)$ and can locally compute $d(y)$ for any $y\in \{0,1,\ldots,|V|\}$.  
\end{lemma}

We will also use the following result by Ghaffari, Kuhn and Su \cite{Ghaffari+PODC17} about randomized routing in networks with small mixing time (see also the discussion in Section 3 of \cite{Chang+PODC19}).
\begin{theorem}[\cite{Ghaffari+PODC17}]\label{th:routing}
In the classical CONGEST model, there exists a $O(\mix(G)\cdot n^{o(1)})$-round algorithm that builds a distributed data structure. This data structure enables the vertices to implement the following routing task with probability at least $1-1/\poly(n)$ in $O(\mix(G)\cdot n^{o(1)})$ rounds: given a set of point-to-point routing requests, each given by the IDs of the corresponding source-destination pair and such that each vertex $u$ is the source and the destination of at most $O(\deg(u))$ messages, delivers all the messages.
\end{theorem}

\noindent{\bf Quantum distributed computing.}
We assume that the reader is familiar with the basic concepts of quantum computation and refer to, e.g., \cite{Nielsen+00} for a good reference. The quantum CONGEST model is defined (see \cite{Elkin+PODC14,LeGall+PODC18} for details) as the quantum version of the classical CONGEST model, where the only difference is that each exchanged message consists of $O(\log n)$ quantum bits instead of $O(\log n)$ bits. In particular, initially the vertices of the network do not share any entanglement.

For the quantum CONGEST model, Le Gall and Magniez~\cite{LeGall+PODC18} introduced a framework for quantum distributed search, which can be seen as a distributed implementation of Grover's search~\cite{GroverSTOC96}, one of the most important centralized quantum algorithms. Let $X$ be a finite set and $f\colon X\to \{0,1\}$ be a Boolean function over~$X$. Let $u$ be an arbitrary vertex of the network (e.g., an elected leader). Assume that vertex~$u$ can evaluate the function $f$ in~$r$ rounds: assume that there exists an $r$-round distributed checking procedure $\Cc$ such that vertex $u$, when receiving as input $x\in X$, outputs $f(x)$.  Now consider the following problem: vertex $u$ should find one element $x\in X$ such that $f(x)=1$ (or report that no such element exists). The trivial strategy is to compute $f(x)$ for each $x\in X$ one by one, which requires $r|X|$ rounds. Ref.~\cite{LeGall+PODC18} showed that in the quantum CONGEST model this problem can be solved with probability at least $1-1/\poly(|X|)$ in $\tilde O(r\sqrt{|X|})$ rounds. 

As explained in \cite{LeGall+PODC18}, the procedure $\Cc$ is often described as a classical (deterministic or randomized) procedure. It can then be quantized using standard techniques: one first transforms it to a reversible map using standard techniques~\cite{Bennett+SICOMP89} and then converts it into a quantum procedure.
%=================================
\section{Reduction to Triangle finding over Subnetworks}\label{section:reduction}
%=================================
In this section we present the reduction by Chang, Pettie and Zhang~\cite{Chang+SODA19} from triangle finding over the whole network to triangle finding over subnetworks with small mixing time. In this section again, $G=(V,E)$ represents the whole network on which we want to solve the triangle finding problem, and we write $n=|V|$.\vspace{2mm}

\noindent{\bf Triangle finding over subnetworks with small mixing time.} We now present the computational problem considered, which we denote $\Tri$ (the description is also summarized in Figure \ref{fig:problem-tri}). 

The input of $\Tri$ is a connected subgraph $\Gin=(\Vin,\Ein)$ of $G$ such that $\mix(\Gin)=\poly(\log n)$, and a set of edges $\Eout\subseteq E$ joining vertices in $\Vin$ to vertices in $V\setminus\Vin$ that satisfies the following condition:
\[
\degin(u)\ge \degout(u) \textrm{ for all } u\in \Vin.
\]
We write $\Vout$ the set of vertices in $V\setminus \Vin$ that appear as an endpoint of an edge in $\Eout$. The goal is to detect if there is a triangle in the subgraph of $G$ induced by the edge set $\Ein\cup\Eout$. Note that such a triangle is either a triangle of~$\Gin$, or consists of two vertices in $\Vin$ and one vertex in $\Vout$. Note that, while $\Gin$ (the subnetwork induced by $\Ein$) has small mixing time, in general nothing can be said about the mixing time of the subnetwork induced by $\Ein\cup\Eout$. \vspace{2mm}

\begin{figure}
\begin{center}
\fbox{
\begin{minipage}{14 cm} \vspace{2mm}

\noindent$\Tri$\\\vspace{-3mm}

\noindent\hspace{3mm} Input: a connected subgraph $\Gin=(\Vin,\Ein)$ of $G$ and a set of edges $\Eout\subseteq E$  

\noindent\hspace{15mm} 
joining vertices in $\Vin$ to vertices in $V\setminus\Vin$

\noindent\hspace{15mm} 
(each vertex $u\in G$ knows if $u\in \Vin$ and gets $\Nn_{\Ein}(u)$ and $\Nn_{\Eout}(u)$)  \vspace{2mm}

%\noindent\hspace{3mm} Active edges: the edges in $\Ein\cup \Eout$
%\vspace{2mm}

\noindent\hspace{3mm} Promise: (i) $\mix(\Gin)=\poly(\log n)$ 

\noindent\hspace{18mm} (ii) $\degin(u)\ge \degout(u)$ for all $u\in \Vin$\vspace{2mm}

\noindent\hspace{3mm} Goal: detect if there exists a triangle $\{u,v,w\}$ with $u,v,w\in \Vin\cup\Vout$ and 

\noindent\hspace{14mm} 
$\{u,v\},\{u,w\},\{v,w\}\in\Ein\cup\Eout$ 
\vspace{2mm}
\end{minipage}
}
\end{center}\vspace{-4mm}
\caption{Problem \Tri.}\label{fig:problem-tri}
\end{figure}

\noindent{\bf The reduction.}
Chang, Pettie and Zhang~\cite{Chang+SODA19} proved that when a good expander decomposition of the network is known, triangle finding over the whole network can be efficiently reduced to solving several instances of $\Tri$. Combined with Theorem~\ref{th:decomposition}, this gives an efficient reduction from triangle finding to $\Tri$, which we state in the following theorem. For completeness we include a sketch of the proof (we refer to  \cite{Chang+SODA19,Chang+PODC19} for the details).
\begin{theorem}[\cite{Chang+SODA19,Chang+PODC19}]\label{th:reduction}
Assume that there exists an $r$-round distributed algorithm $\Aa$ that solves the problem $\Tri$ with probability at least $1-1/n^3$ and uses only the edges in $\Ein\cup\Eout$ for communication. Then there exists a $O(r\log n+n^{0.1})$-round distributed algorithm that solves the triangle finding problem over the whole graph $G=(V,E)$ with probability at least $1-1/\poly(n)$.
\end{theorem}
\begin{proof}[Sketch of the proof]
The first step of the reduction computes a good decomposition of the whole network $G$ in $O(n^{0.1})$ rounds using Theorem \ref{th:decomposition}. Let $V=V_1\cup V_2\cup\cdots\cup V_s$, for some integer $s\le n$, denote the decomposition and $\Einter$ denote the set of inter-component edges. By definition, the set $\Einter$ satisfies $|\Einter|\le 0.1|E|$.

For any index $i\in[s]$, let us write $G_i=(V_i,E_i)$ the subgraph of $G$ induced by $V_i$. We say that a vertex $u\in V_i$ is good if $\deg_{E_i}(u)\ge \deg_{\Einter}(u)$; otherwise we say that $u$ is bad. Let $\Einter_i$ be the set of edges in $\Einter$ that are adjacent to a good vertex of $G_i$. Let $\Enew_i$ be the set of edges in $E_i$ that are adjacent to a bad vertex of $G_i$. Define the set 
\[
\Enew=\Einter\cup\Enew_1\cup\cdots\cup\Enew_s
\]
and observe that $|\Enew|\le |\Einter|+2|\Einter|\le 0.3|E|$.

The triangles in $G$ can be classified into the following four types.
\begin{itemize}
\item
Type 1: triangles with three vertices in a same component $G_{i}$.
\item
Type 2: triangles with two vertices in a same component $G_i$ and the third vertex in another component $G_j$, in which the two vertices in $V_i$ are good.
\item 
Type 3: triangles with two vertices in a same component $G_i$ and the third vertex in another component $G_j$, in which at least one of the two vertices in $V_i$ is bad.
\item
Type 4: triangles with three vertices in distinct components.
\end{itemize}

For each index $i\in[s]$, we use Algorithm $\Aa$ to solve $\Tri$ on instance $(\Gin,\Eout)$ with $\Gin=G_i$ and $\Eout=\Einter_i$. A crucial point is that this can be done for all $i$'s in parallel by ``doubling'' the bandwidth (i.e., by using~$2r$ rounds in total), since Algorithm~$\Aa$ on instance $(\Gin,\Eout)$ only uses the edges in $E_i\cup\Einter_i$ for communication. Indeed, the communication networks are disjoint for all instances, except for the intercomponent edges that can be shared by two instances. This detects all the triangles of types 1 and 2 in the graph. 

Another crucial observation is that all remaining potential triangles (i.e., the triangles of type~3 and~4) have their three edges contained in the set $E^\textrm{new}$. It is thus enough to recurse on this set, i.e., to repeat the same methodology with $E$ replaced by $\Enew$. Since $|\Enew|\le 0.3|E|$, after $O(\log n)$ levels of recursion the algorithm finishes. The overall complexity of this second part is thus $O(r\log n)$ rounds. 
\end{proof}

%===========================================
%===========================================
\section{Main Quantum Algorithm}\label{section:qa}
%===========================================
%===========================================
In the classical CONGEST model, Chang, Pettie and Zhang \cite{Chang+SODA19} have shown that the problem $\Tri$ can be solved with high probability in $\tilde{O}(n^{1/3})$ rounds using only the edges in $\Ein\cup\Eout$ for communication, which leads to a $\tilde O(n^{1/3})$-round classical algorithm for triangle finding via Theorem \ref{th:reduction}. Our main technical result is the following theorem.
\begin{theorem}\label{th:sub}
In the quantum CONGEST model, there is a $\tilde O(n^{1/4})$-round quantum algorithm that solves the problem $\Tri$ with probability at least $1-1/n^3$ and uses only the edges in $\Ein\cup\Eout$ for communication.
\end{theorem}
Theorem \ref{th:main} then immediately follows from Theorem \ref{th:reduction} and Theorem \ref{th:sub}.

The goal of this section is to prove Theorem \ref{th:sub}. 
For brevity we write $\Vv=\Vin\cup\Vout$ and $\Ee=\Ein\cup\Eout$.
We also write $\myn =|\Vv|$ and $\mym=|\Ee|$. Note that $\mym\ge \myn/2$ since the graph $(\Vv,\Ee)$ is connected. Let $S\subseteq \Vv$ be the subset of all vertices $u\in \Vv$ such that
\[
\deg_\Ee(u)\ge \mym/\sqrt{\myn}.
\] 
Observe that $|S|\le 2\sqrt{\myn}$. 

In Section \ref{sub-large} below we present a $\tilde O(n^{1/4})$-round quantum algorithm that detects the existence of a triangle containing at least one vertex from~$S$. We then describe, in Section~\ref{sub-small}, our main technical contribution: a $\tilde O(n^{1/4})$-round quantum algorithm that detects the existence of a triangle under the assumption $S=\emptyset$. The quantum algorithm of Theorem~\ref{th:sub} then follows by combining these two algorithms, since (as already observed in prior works~\cite{Chang+SODA19, Chang+PODC19}) detecting whether there exists a triangle with no vertex in $S$ reduces to the case $S=\emptyset$.

Let us provide some explanations about why detecting the existence of a triangle with no vertex in $S$ reduces to the case $S=\emptyset$. One natural approach is to focus on the subgraph where all the nodes in $S$ are removed. This approach nevertheless does not immediately work since this may make the graph disconnected and may significantly reduce the number of edges (in which case Lemma~\ref{lemma} may not anymore be applicable). Instead, we now briefly describe a method that keeps the graph connected and the number of edges (almost) unchanged. The idea is simply to ``virtually'' replace each node $u\in S$ by a star of degree $\sqrt{\myn}$ and spread the $\deg_\Ee(u)$ incident edges of $u$ evenly into the leaves of the star so that each leaf has $\deg_\Ee(u)/\sqrt{\myn}<\mym/\sqrt{\myn}$ incident edges. Since $|S|\le 2\sqrt{\myn}$, this can be done by introducing only $\Theta(\myn)$ virtual nodes.

%%%%%%%%%%%%%%%%%%%%%%%%%%%%%%%%%%%%%%%%%%%%%%%%%%%%%
\subsection{Finding a triangle containing (at least) one high-degree vertex}\label{sub-large}
%%%%%%%%%%%%%%%%%%%%%%%%%%%%%%%%%%%%%%%%%%%%%%%%%%%%%
In this subsection we describe how to detect, in $\tilde O(n^{1/4})$ rounds, the existence of a triangle with edges in $\Ee$ that contains at least one vertex from~$S$. We will use the following lemma, which is a straightforward application of the framework for distributed quantum search described in Section \ref{section:preliminaries}, but can also be seen as an adaptation of the quantum protocol by Buhrman, Cleve and Wigderson~\cite{Buhrman+STOC98} for the disjointness function in two-party quantum communication complexity. 
\begin{lemma}\label{lemma:qsearch}
Consider any two adjacent vertices $u$ and $v$, each owning a set $T_u\subseteq V$ and a set $T_v\subseteq V$, respectively. There is a quantum algorithm that checks if $T_u\cap T_v\neq \emptyset$ with high probability in $\tilde O(\sqrt{\min\{|T_u|,|T_v|\}})$ rounds. Moreover, this algorithm only uses communication along the edge $\{u,v\}$.
\end{lemma}
\begin{proof}
Consider the subnetwork consisting only of the two vertices $u$ and $v$ and the edge $\{u,v\}$.
Set $X=T_u$ and define the function 
\[
f(x) = \left\{
\begin{array}{ll}
1&\textrm{ if } x\in T_v,\\
0&\textrm{ otherwise,}
\end{array}
\right. 
\]
for any $x\in X$. Obviously, for any $x\in X$, vertex $u$ can compute the value $f(x)$ in 2 rounds. We can thus apply the quantum distributed search framework of Section \ref{section:preliminaries} with vertex $u$ acting as a leader, which enables $u$ to check whether there exists $x\in X$ such that $x \in T_v$ in $\tilde O(\sqrt{|T_u|})$ rounds.

By symmetry there also exists a quantum algorithm that enables vertex $v$ to decide whether $T_u\cap T_v\neq \emptyset$ in $\tilde O(\sqrt{|T_v|})$ rounds. Combining these two algorithms gives the claimed complexity.
\end{proof}

We now explain how our quantum algorithm works. First of all, observe that since each vertex~$u$ receives as input $\Nn_{\Ein}(u)$ and $\Nn_{\Eout}(u)$, each vertex knows whether it is in~$S$ or not. Each vertex first tells this to all its neighbors. This requires 1 round of communication. 

Each vertex $u\in \Vv$ then computes, locally, the set 
\[
T_u=\Nn_\Ee(u) \cap S=(\Nn_{\Ein}(u)\cup \Nn_{\Eout}(u))\cap S. 
\]
Note that for each edge $\{u,v\}\in\Ee$, there exists $w\in S$ such that $\{u,v,w\}$ is in a triangle with three edges in $\Ee$ if and only if $T_u\cap T_v\neq \emptyset$. Thus, for each edge $\{u,v\}\in\Ee$, the vertices $u$ and $v$ use the quantum algorithm of Lemma \ref{lemma:qsearch} to decide whether $T_u\cap T_v\neq \emptyset$ or not. Since this algorithm only communicates through the edge $\{u,v\}$, it can be applied in parallel to all edges $\{u,v\}\in\Ee$. This gives overall round complexity $\tilde O(\sqrt{|S|})=\tilde O(n^{1/4})$.

%%%%%%%%%%%%%%%%%%%%%%%%%%%%%%%%%%%%%%%%%%
\subsection{Finding a triangle with only low-degree vertices}\label{sub-small}
%%%%%%%%%%%%%%%%%%%%%%%%%%%%%%%%%%%%%%%%%%
In this subsection we assume that the inequality 
\[
\deg_\Ee(u)< \mym/\sqrt{\myn}.
\] 
holds for all vertices $u\in \Vv$, i.e., we assume that $S=\emptyset$. We show how to detect, in $\tilde O(n^{1/4})$ rounds,  the existence of a triangle with edges in $\Ee$ in this case as well.\vspace{2mm}

\noindent{\bf Partitioning the set $\boldsymbol{\Vv}$.}
Let us write $t=\floor{\sqrt{\myn}/(30\log \myn)}$. We randomly partition the set $\Vv$ into~$t$ subsets $\Vv_1,\ldots,\Vv_t$ as follows: each vertex $u\in \Vv$ selects an integer $i$ uniformly at random in the set $[t]$ and joins the set $\Vv_i$. Vertex $u\in \Vv$ then tells its neighbors the value $i$, which can be done in 1 round. Each vertex therefore learns in which sets its neighbors have been included. 

For any $i,j\in[t]$, let $E(\Vv_i,\Vv_j)$ denote all the edges in $\Ee$ with one endpoint in $\Vv_i$ and one endpoint in $\Vv_j$. Our analysis will rely on the following lemma.
\begin{lemma}\label{lemma:sparse}
With probability $1-1/\poly(n)$, the following statement is true: for all $i,j\in[t]$,
\[
|E(\Vv_i,\Vv_j)|=O\left(\frac{\mym (\log \myn)^2}{\myn}\right).
\]
\end{lemma}
\begin{proof}
Let us first consider the case $i=j$. We apply Lemma \ref{lemma} over the graph generated by the vertex set $\Vv$ and using the probability $p=1/t$. Note that 
\[
p^2\mym=\frac{\mym}{(\floor{\sqrt{\myn}/(30\log\myn)})^2}\ge 
\frac{\myn/2}{({\sqrt{\myn}/(30\log\myn)})^2}\ge
400 (\log \myn )^2
\]
and
\[
\frac{\mym p}{20\log \myn}=\frac{\mym}{20\log \myn\floor{\sqrt{\myn}/(30\log\myn)}}
\ge \frac{\mym}{\sqrt{\myn}},
\]
which implies that the two conditions in Lemma \ref{lemma} are satisfied.

In the case $i\neq j$, we apply Lemma \ref{lemma} over the graph generated by the vertex set $\Vv$ again, but using the probability $p=2/t$. The conclusion is the same.
\end{proof}

\noindent{\bf Partitioning the triples of indices.}
Let us write $\Lambda$ the set of all triples $(i,j,k)$ with $i,j,k\in[t]$, i.e., $\Lambda = [t]\times  [t]\times [t]$. Let us partition this set into $t$ sets $\Lambda_1,\ldots,\Lambda_t$, each containing $t^2$ triples, as follows. For each $\ell\in[t]$, define the set $\Lambda_\ell$ as:
\[
\Lambda_\ell=\big\{(i,j,1+(i+j+\ell\bmod t))\:|\: (i,j)\in[t]\times [t])\big\}.
\]
Our analysis will rely on the following lemma, which immediately follows from the definition of the sets $\Lambda_\ell$.

\begin{lemma}\label{lemma:unique}
The following statements are true for all $\ell\in[t]$ and all triples $(i,j,k)\in \Lambda_\ell$: 
\begin{itemize}
\item
there is no index $i'\in[t]\setminus\{i\}$ such that $(i',j,k)\in \Lambda_\ell$;
\item
there is no index $j'\in[t]\setminus\{j\}$ such that $(i,j',k)\in \Lambda_\ell$;
\item
there is no index $k'\in[t]\setminus\{k\}$ such that $(i,j,k')\in \Lambda_\ell$.
\end{itemize}
\end{lemma}

\noindent{\bf Assigning the triples to vertices.}
For each $\ell\in[t]$, we assign the $t^2=\Theta(\myn/(\log \myn)^2)$ triples in~$\Lambda_\ell$ to the vertices in $\Vin$. The assignment should be made carefully, so that each vertex is assigned a number of triples proportional to its degree and, additionally, all the vertices know to which vertex each triple in $\Lambda_\ell$ is assigned. To achieve this goal we use the same approach as in \cite{Chang+SODA19}, which is based on the ID assignment of Lemma~\ref{th:ID}.

We first apply Lemma \ref{th:ID} to the subnetwork $\Gin$ in order to obtain an ID assignment $\gamma\colon \Vin\to\{1,\ldots,|\Vin|\}$ and the degree estimator function \[d\colon \{1,\ldots,|\Vin|\}\to \{0,1,\ldots,\floor{\log _2 |\Vin|}\}\] satisfying the properties in the lemma. This requires $O(\diam(\Gin)+\log n)=O(\mix(\Gin)+\log n)=\poly(\log n)$ rounds. For any vertex $u\in \Vin$, define the quantity
\[
r_u=
\frac{
2^{d(\gamma(u))}}
{\mym/\myn}.
\]
Note that since $d(\gamma(u))=\floor{\log_2(\degin(u))}$, the quantity $r_u$
is an approximation of the ratio between $\degin(u)$ and the average degree of the subgraph $(\Vv, \Ee)$. Observe that 
\begin{align*}
\sum_{u\in \Vin} r_u \ge
 \sum_{u\in \Vin} \frac{
\degin(u)/2}
{\mym/\myn}=
\frac{|\Ein|}
{\mym/\myn}\ge \myn/2,
\end{align*}
since $|\Ein|\ge \mym/2$.
Now define the quantity
\[
q_u=\left\{
\begin{array}{cc}
0&\textrm{ if } r_u\le 1/4,\\
\ceil{r_u}&\textrm{otherwise},
\end{array}
\right.
\]
and observe that 
\begin{equation}\label{eq:sum}
\sum_{u\in \Vin} q_u \ge 
\sum_{u\in \Vin} r_u \:\:- \frac{|\Vin|}{4} 
\ge
\frac{\myn}{4}\ge
t^2.
\end{equation}

We can now explain the assignment of the triples from $\Lambda_\ell$. We fix an arbitrary order (known to all the vertices of the network) on the triples of each $\Lambda_\ell$. For concreteness, let us choose the lexicographic order. We start by assigning to the vertex $u_1\in \Vin$ such that $\gamma(u_1)=1$ the first $q_{u_1}$ triples of $\Lambda_\ell$ in the lexicographic order,  then assign to the vertex $u_2\in \Vin$ such that $\gamma(u_2)=2$ the next $q_{u_2}$ triples from $\Lambda_\ell$ in the lexicographic order, and repeat the process until all the triples of~$\Lambda_\ell$ have been assigned (Equation (\ref{eq:sum}) guarantees that all triples are assigned by this process). For each vertex $u\in\Vin$, let us write $\Lambda_\ell^u\subseteq \Lambda_\ell$ the set of triples assigned to $u$.

A crucial observation is that each vertex of the network can locally compute, for any $\ell\in[t]$ and any triple $(i,j,k)\in\Lambda_\ell$, the ID of the vertex to which $(i,j,k)$ is assigned, since each vertex knows the value $d(y)$ for all $y\in\{0,1\ldots,|\Vin|\}$.\vspace{2mm}

%%%%%%%%%
\noindent{\bf Description of the quantum search algorithm.}
Consider the function 
\[
f\colon [t] \to \{0,1\}
\]
defined as follows. For any $\ell\in[t]$, we have $f(\ell)=1$ if and only if there exists a triple $(i,j,k)\in \Lambda_\ell$ such that the graph $G$ has a triangle with one vertex in the set $\Vv_i$, one vertex in~$\Vv_j$, one vertex in~$\Vv_k$ and its three edges in $\Ee$. Our quantum algorithm implements the quantum distributed search framework described in Section \ref{section:preliminaries} with $X=[t]$ to detect if there exists one index $\ell\in[t]$ such that $f(\ell)=1$. This approach obviously detects the existence of a triangle with three edges in $\Ee$, i.e., it solves our problem. The complexity of this approach, as explained in Section \ref{section:preliminaries}, is $\tilde O(\sqrt{t}\delta)=\tilde O(n^{1/4}\delta)$ rounds, where~$\delta$ is the round complexity of the checking procedure. We present below a checking procedure with round complexity $\delta=\tilde O(\mix(\Gin))$. Since $\mix(\Gin)=\poly(\log n)$ from our assumption on $\Gin$, the overall round complexity is $\tilde O(n^{1/4})$, as claimed.\vspace{2mm}

%%%%%%%%%
\noindent{\bf Description of the checking procedure.}
We now describe a classical randomized checking procedure that enables the leader, on an input $\ell\in[t]$, to evaluate the value $f(\ell)$. As explained in Section \ref{section:preliminaries} such a classical procedure can then be converted into a quantum procedure using standard techniques. In the checking procedure, the leader first broadcasts the information ``$\ell$'' to all the vertices of the network. This can be done in $\diam(\Gin)\le\mix(\Gin)$ rounds. Then each vertex $u\in\Vin$ checks, for each $(i,j,k)\in\Lambda_\ell^u$, whether there exists a triangle with one vertex in $\Vv_i$, one vertex in $\Vv_j$, one vertex in~$\Vv_k$ with three edges in $\Ee$. In order to do that, vertex~$u$ simply needs to collect all the edges in $E(\Vv_i,\Vv_j)\cup E(\Vv_i,\Vv_k)\cup E(\Vv_j,\Vv_k)$ for each $(i,j,k)\in\Lambda_\ell^u$. From Lemma~\ref{lemma:sparse} and from the definition of the set $\Lambda_\ell$, this requires
\[
\tilde O\left(
\frac{\mym}{\myn}\times 
q_u
\right)=
\tilde O\left(
\frac{\mym}{\myn}\times 
\floor{r_u}
\right)=
\tilde O\left(
\degin(u)
\right)
\]
incoming messages. Conversely, let us consider the number of outgoing messages needed to gather the information about the edges. Lemma \ref{lemma:unique} guarantees that the information about each edge only needs to be communicated to one vertex,  
which implies that each vertex $u$ is the source of $\degin(u)$ messages.
Theorem \ref{th:routing} thus implies that the checking procedure can be implemented in $O(\mix(\Gin)\cdot n^{o(1)})$ rounds. The leader then checks if one of the vertices in $\Vin$ found a triangle, which can be done in $O(\diam(\Gin))=O(\mix(\Gin))$ rounds. 

In order to reduce the complexity from $O(\mix(\Gin)\cdot n^{o(1)})$ to $\tilde O(\mix(\Gin))$ we simply need to modify slightly the routing scheme from \cite{Ghaffari+PODC17}, exactly as done in the classical case (see Section 3 of~\cite{Chang+PODC19}). 

\section*{Acknowledgments}
The authors are grateful to anonymous reviewers for helpful comments.
TI was partially supported by JST SICORP grant No.~JPMJSC1606 and JSPS KAKENHI grant No.~JP19K11824.
FLG was supported by JSPS KAKENHI grants Nos.~JP15H01677, JP16H01705, JP16H05853, JP19H04066 and by the MEXT Quantum Leap Flagship Program (MEXT Q-LEAP) grant No.~JPMXS0118067394. 
FM was partially supported by the ERA-NET Cofund in Quantum Technologies project QuantAlgo and the French ANR Blanc project QuData.

%=====================================
%=====================================
%\bibliographystyle{plain}
%\bibliography{refs}
%=====================================

\end{document}